\documentclass[runningheads]{llncs}
\usepackage[T1]{fontenc}

\usepackage{amsmath}
\usepackage{amssymb}
\usepackage{amsfonts}
\usepackage{verbatim}

\usepackage{stmaryrd}
\usepackage{tikz}
\usepackage{adjustbox}
\usepackage{caption}
\usepackage{subcaption}
\usepackage{enumerate}
\usepackage{hyperref}

\usepackage{tikz}
\usetikzlibrary{automata, shapes,arrows}
\usetikzlibrary{calc,positioning}
\usetikzlibrary{backgrounds}

\usepackage{graphicx}
\newcommand{\head}{\text{head}}
\newcommand{\tail}{\text{tail}}
\newcommand{\turn}{\text{turn}}
\newcommand{\move}{\text{move}}
\newcommand{\routing}{\text{routing}}
\newcommand{\Z}{\mathbb{Z}}
\renewcommand{\ll}{\llbracket}
\newcommand{\rr}{\rrbracket}
\newcommand{\cR}{{\mathcal{R}}}

\begin{document}
\title{Generalized ARRIVAL Problem for  Rotor Walks in Path Multigraphs }
%
%
\author{David Auger\inst{1} \and
Pierre Coucheney\inst{1} \and Loric Duhazé\inst{1}\and
Kossi Roland Etse\inst{1}}
\authorrunning{D. Auger et al.}
%

\institute{
DAVID Lab., UVSQ, Université Paris Saclay, 45 avenue des Etats-Unis,78000,Versailles, France
}

%
\maketitle              
\begin{abstract}
Rotor walks are cellular automata that determine deterministic traversals of particles in a directed multigraph using simple local rules, yet they can generate complex behaviors. Furthermore, these trajectories exhibit statistical properties similar to random walks.

In this study, we investigate a generalized version of the reachability problem known as {\sc arrival} in Path Multigraphs, which involves predicting the number of particles that will reach designated target vertices. We show that this problem is in NP and co-NP in the general case. However, we exhibit algebraic invariants for Path Multigraphs that allow us to solve the problem efficiently, even for an exponential configuration of particles. These invariants are based on harmonic functions and are connected to the decomposition of integers in rational bases.
\end{abstract}

%
%

\section{Introduction}

The \emph{rotor routing}, or \emph{rotor walk model}, has been studied under different names: \emph{eulerian walkers} \cite{priezzhev1996eulerian,povolotsky1998dynamics} and \emph{patrolling algorithm} \cite{yanovski2003distributed}. It shares many properties with a more algebraically focused model: \emph{abelian sandpiles} \cite{bjorner1991chip,Holroyd2008}. General introductions to this cellular automaton can be found in \cite{giacaglia2011local} and \cite{Holroyd2008}.

Here is how a rotor walk works: in a directed graph, each vertex $v$ with an outdegree of $k$ has its outgoing arcs numbered from $1$ to $k$. Initially, a particle is placed on a starting vertex, and the following process is repeated. On the initial vertex, the particle moves to the next vertex following arc $1$. The same rule then applies on subsequent vertices. However, when a vertex is revisited, the particle changes its movement to the next arc, incrementing the number until the last arc is used. Then, the particle restarts from arc $1$ if it visits this vertex again.

This simple rule defines the rotor routing, which exhibits many interesting properties. Particularly, if the graph is sufficiently connected, the particle will eventually reach certain target vertices known as sinks. The time required for such exploration can be exponential in the number of vertices. The problem of determining, given a starting configuration (numbering) of arcs and an initial vertex, which sink will be reached first, is known as the ARRIVAL problem. It was defined in \cite{dohrau2017arrival}, along with a proof that the problem belongs to the complexity class NP~$\cap$~co-NP. Although the problem is not known to be in P, \cite{gartner2018arrival} showed that it belongs to the smaller complexity class UP~$\cap$~co-UP. Furthermore, a subexponential algorithm based on computing a Tarski fixed point was proposed in \cite{gartner_et_al:LIPIcs.ICALP.2021.69}.

Despite these general bounds, little is known about efficiently solving the problem in specific graph classes, especially when extending it to the routing of multiple particles. In \cite{auger2022polynomial}, we addressed the problem in multigraphs with a tree-like structure and provided a linear algorithm for solving it with a single particle. However, the recursive nature of the algorithm provided limited insights into the structure of rotor walks in the graph. We also examined the structure of rotor walks and the so-called sandpile group in the case of a simple directed path, where simple invariants can explain the behavior of rotor walks.

In this work, we focus specifically on a family of multigraphs that consist of directed paths with a fixed number of arcs going left and right on each vertex, with a sink located at both ends of the path. We present an efficient algorithm for solving the ARRIVAL problem in this general context, considering a potentially exponential number of particles and antiparticles, a concept introduced in \cite{giacaglia2011local}. Our approach involves introducing algebraic invariants for rotor walks and chip-firing, enabling a complete description of the interplay between particle configurations and rotor configurations/walks. These invariants are derived from harmonic functions in graphs, which are functions invariant under chip-firing. Additionally, we introduce a related concept for rotor configurations called arcmonic functions, inspired by \cite{hoang2022two}.

An essential tool for analyzing rotor routing in Path Multigraphs is the decomposition of integer values, which is closely associated with the AFS number system (\cite{frougny2012rational}), where numbers are decomposed into rational bases. While we draw inspiration from these results, our approach focuses on proving precisely what is necessary, using our own methodology.

Additionally, we derive other outcomes, such as the cardinality of the Sandpile Group of Path Multigraphs or its cyclic structure. These results can also be derived from Kirchoff's Matrix-Tree Theorem or the notion of co-eulerian graphs~\cite{farrell2016coeulerian}. Nevertheless, our results remain self-contained.

\section{Mechanics and Tools for Rotor Routing in Multigraphs}

\subsection{Multigraphs}
\label{sec.graph_def}

A \textbf{directed multigraph} $G$ is a tuple $G=(V,A,\head,\tail)$
where $V$ and $A$ are respectively finite sets of \emph{vertices} and \emph{arcs}, and
{\it head} and {\it tail} are maps from $A$ to $V$ defining
incidence between arcs and vertices. An arc with tail $x$ and head $y$ is said to be from $x$ to $y$.
Note that multigraphs can have multiple arcs with the same head and tail,  as well as loops.

For a vertex $u\in V$, we denote by $A^{+}(u)$ the subset of arcs going out of $u$, i.e. $A^{+}(u)=\{a\in A\mid \tail(a)=u\}$ and $\deg^+(u) = |A^+(u)|$ is the outdegree of $u$.  We  denote by $V_0$ the set of vertices with positive outdegree and $S_0$ vertices with zero outdegree, i.e. {\bf sinks}.
A directed multigraph is \textbf{stopping} if for every vertex $u$, there is a directed path from $u$ to a sink. {\bf In this whole paper, we suppose that $G$ is a stopping multigraph.}

In the second part of this work, we consider the following multigraph:
the {\bf Path multigraph } $P^{x,y}_n$ on $n+2$ vertices is a multigraph $G= (V_0 \cup S_0,A,\head,\tail)$  such that:
\begin{itemize}
\item  $V_{0}=\{u_{1},u_{2},...,u_{n}\}$ and $S_{0}=\{u_{0},u_{n+1}\}$;
\item  for $k\in \ll 1,n \rr$, we have $\deg^+(u_k)=x+y$  with  $x$ arcs from $u_k$ to $u_{k+1}$ and  $y$ arcs from $u_k$ to $u_{k-1}$
\item $u_0$ and $u_{n+1}$ are considered as {\bf sinks} with no outgoing arcs.
\end{itemize}

This graph is clearly stopping if $x+y \geq 1$.
See Fig. \ref{fig:ex_graph}  for a representation of $P^{2,3}_n$.

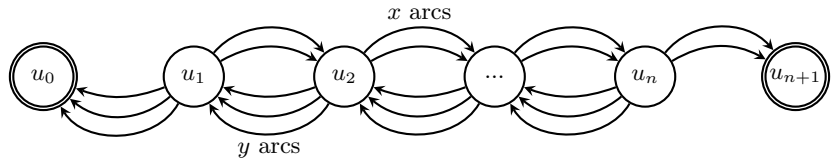
\begin{figure}[htbp]
     \centering
\begin{tikzpicture}[>=stealth, auto, node distance=2cm, thick]

\tikzset{
  state/.style={circle, draw, inner sep=0pt, minimum size=8mm}
}

\node[state] (01) {$u_0$};
\node[state] (1) [right of=01] {$u_1$};
\node[state] (2) [right of=1] {$u_2$};
\node[state] (3) [right of=2] {$...$};
\node[state] (4) [right of=3] {$u_n$};
\node[state] (51) [right of=4] {$u_{n+1}$};

\node[shape=circle,draw=black] (0) at (01) {$~~~~~~$};
\node[shape=circle,draw=black] (5) at (51) {$~~~~~~$};


\path[->]
  (1) edge[bend left=20] (0)
  (1) edge[bend left=40] (0)
  (1) edge[bend left=60] (0)
  (1) edge[bend left=30] (2)
  (1) edge[bend left=50] (2)

  (2) edge[bend left=20] (1)
  (2) edge[bend left=40] (1)
  (2) edge[bend left=60] node[below] {$y$ arcs} (1)
  (2) edge[bend left=30] (3)
  (2) edge[bend left=50] node[above] {$x$ arcs} (3)

  (3) edge[bend left=20] (2)
  (3) edge[bend left=40] (2)
  (3) edge[bend left=60] (2)
  (3) edge[bend left=30] (4)
  (3) edge[bend left=50] (4)

  (4) edge[bend left=20] (3)
  (4) edge[bend left=40] (3)
  (4) edge[bend left=60] (3)
  (4) edge[bend left=30] (5)
  (4) edge[bend left=50] (5);

\end{tikzpicture}
\caption{\footnotesize The Path Multigraph $P^{2,3}_n$.}
\label{fig:ex_graph} 
\end{figure}

We consider the case $n \geq 1$, and $1 \leq x < y$ with $x,y$ coprime.

\subsection{Rotor Structure}

If $u \in V_0$, a \textbf{rotor order} at $u$  is an operator denoted by $\theta_{u}$
such that:
\begin{itemize}
\item $\theta_{u}: A^{+}(u) \rightarrow A^{+}(u)$ ;
\item for all $a\in A^{+}(u)$, the orbit $\{a,\theta_{u}(a),\theta_{u}^{2}(a),...,\theta_{u}^{\deg^+(u)-1}(a)\}$
of $a$ under $\theta_{u}$ is equal to $A^+(u)$,
where $\theta_{u}^{k}(a)$ is the composition of $\theta_{u}$ applied
to arc $a$ exactly $k$ times. 
\end{itemize}

A {\bf rotor order} for $G$ is then a map $\theta : A \rightarrow A$ such that the restriction $\theta_u$ of $\theta$ to $A^+(u)$
is a rotor order at $u$ for every $u \in V_0$. Note that all $\theta_u$ as well as $\theta$ are one to one. 
If $C \subseteq V_0$, the composition of operators $\theta_u$ for all $u \in C$ does not depend on the order of composition since they act on disjoint sets $A^+(u)$; we denote by $\theta_C$ this operator and $\theta^{-1}_C$ is its inverse.
Finally, we use the term {\bf rotor graph} to denote a stopping multigraph together with a rotor order $\theta$.

In $P^{x,y}_n$, we define a rotor order by simply considering all arcs going right before all arcs going left, cyclically (see Fig. \ref{fig:rotororder}). Formally, let $a^k_i$ denote for $i \in \ll 0, x-1\rr$ the $x$ arcs from $u_k$ to $u_{k+1}$ and for $i \in \ll x,x+y-1\rr$ the $y$ arcs from $u_k$ to $u_{k-1}$; then we define
\[\theta(a^k_i) = a^k_j \text{ with } j = i + 1 \mod x+y.\]

\begin{figure}[htbp]
\centering
\begin{tikzpicture}[>=stealth, auto, node distance=2.5cm, thick]

\tikzset{
  state/.style={circle, draw, inner sep=0pt, minimum size=8mm}
}

\node[state] (g) {$u_{k-1}$};
\node[state,right of= g] (m) {$u_k$};
\node[state,right of= m] (d) {$u_{k+1}$};

\draw [ thick,->,>=stealth](1.5,1) arc (0:330:0.4cm);

\draw[->] (m) to[bend left=10] node[above] {$a^k_x$} (g);
\draw[->,dashed] (m) to[bend left=30] node[above] {} (g);
\draw[->,dashed] (m) to[bend left=50] node[above] {} (g);
\draw[->] (m) to[bend left=70] node[below] {$a^k_{x+y-1}$} (g);

\draw[->] (m) to[bend left=80] node[above] {$a^k_{x-1}$} (d);
\draw[->,dashed] (m) to[bend left=60] node[below] {} (d);
\draw[->,dashed] (m) to[bend left=40] node[below] {} (d);
\draw[->] (m) to[bend left=20] node[below] {$a^k_0$} (d);

\end{tikzpicture}
\caption{\footnotesize  Rotor order at a vertex $u_k$
in the Path Multigraph $P^{x,y}_n$  }
\label{fig:rotororder}

\end{figure}
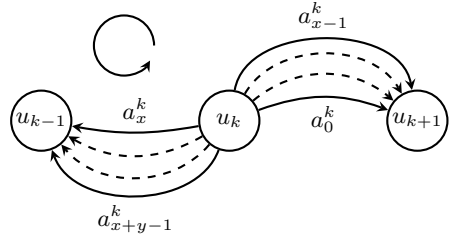

\subsection{Configurations}

\begin{definition}
A \textbf{rotor configuration}  of a rotor graph $G$
is a mapping $\rho$ from $V_{0}$ to $A$ such that $\rho(u)\in A^+(u)$ for all $u\in V_{0}$.  We denote by $\cR(G)$ or simply $\cR$ the set of all rotor configurations of the rotor graph $G$.
\end{definition}

The \textbf{graph induced} by $\rho$ on $G=(V,A,\head,\tail)$ is  
\[G(\rho) = (V,\rho(V_0),\head,\tail),\] in which each vertex in $V_0$ has outdegree one.

\begin{definition}
A \textbf{particle configuration}  of a rotor graph $G$
is a mapping $\sigma$ from $V$ to $\Z$. We denote by $\Sigma(G)$ or simply $\Sigma$ the set of all particle configurations of the rotor graph $G$.
\end{definition}

The set $\Sigma(G)$ can be identified with $\Z^V$ and has a natural structure of additive abelian group. If $u \in V$, we identify $u$ with the element of $\Sigma(G)$
with exactly one chip on $u$. Thus we can write, e.g. $\sigma + 3u$ to denote the configuration obtained from $\sigma \in \Sigma$ by adding 3 to $\sigma(u)$.

If $\sigma(u) \geq 0$, we interpret it as a number of particles
 on vertex $u$, whereas if $\sigma(u) \leq 0$ it can be interpreted as antiparticles, or simply a debt of particles. The {\bf degree} of a particle configuration $\sigma$ is defined by $\deg(\sigma) = \sum_{u \in V} \sigma(u)$.

Finally, a {\bf rotor-particle} configuration is an element of $\cR(G) \times \Sigma(G)$.

\subsection{Rotor Routing}

\begin{definition}
Let $G$ be a rotor graph, we define operators indexed by vertices $u \in V_0$ on $\cR(G) \times \Sigma(G)$:
\begin{itemize}
\item $\move^+_u:\cR(G) \times \Sigma(G) \rightarrow \cR(G) \times \Sigma(G)$ is defined by
\[\move^+_u(\rho,\sigma)=(\rho, \sigma + \head(\rho(u)) - u);\]
\item $\turn^+_u: \cR(G) \times \Sigma(G) \rightarrow \cR(G) \times \Sigma(G)$ is defined
by $\turn^+_u(\rho,\sigma)=(\theta_u \circ \rho,\sigma)$.
\end{itemize}
\end{definition}

Note that $\theta_u \circ \rho$ is the rotor configuration equal to $\rho$ on all vertices except in $u$ where $\theta$ has updated the arc.
Applying $\move^+_u$ to $(\rho,\sigma)$ can be interpreted as moving a particle from $u$ to the head of arc $\rho(u)$, whereas applying $\turn^+_u$ updates the rotor configuration at $u$. 
It is easy to see that these operators are bijective on $\cR(G) \times \Sigma(G)$, and we denote by $\move^-_u$ and $\turn^-_u$ their inverses.

We now define the routing operators by $\routing^+_u = \turn^+_u \circ \move^+_u$, and its inverse is obviously $\routing^-_u = \move^-_u \circ  \turn^-_u $. Routing a rotor-particle  configuration $(\rho,\sigma)$ consists in applying a series of 
$\routing^+$ and $\routing^-$ operators. Since they act on different vertices and disjoint sets of arcs, the following result is straightforward. 

\begin{lemma}
    The family of operators $\routing^+_u$ and $\routing^-_u$ for all $u \in V_0$  commute.
\end{lemma}

Since the order in which routing operators are applied does not matter, we define a {\bf routing vector} as a map from $V_0$ to $\Z$. We define $\routing^r$ as the operator obtained by composing all elements of the family 
\[ \{ (\routing^+_{u})^{r(u)} \}_{u \in V_0} \]
in any order, where the exponent $r(u)$ stands for composition of the operator or its inverse with itself, depending on the sign of $r(u)$. We shall use the term \emph{routing} when we apply any operator $\routing^r$ as well.

We end this subsection by pointing out that the kind of routing defined here, which we call {\it move and turn routing}, is used in \cite{dohrau2017arrival}  and \cite{gartner2018arrival}, and is more adapted to study the {\sc arrival} problem. Another kind of routing, the {\it turn and move routing}, used for instance in \cite{Holroyd2008,giacaglia2011local}, is more widely used in the literature and is more adapted to study the link between the sandpile group and rotor configurations. However, it is easy to see that these two definitions of routing are conjugate by $\theta$, and all results obtained for one of them can be translated into the other context.

\subsection{Legal Routing and \sc arrival}

 Applying $\routing^+_u$ to $(\rho,\sigma) \in \cR \times \Sigma$ is said to be a {\bf legal routing} if $\sigma(u) > 0$. A sequence of legal routings 
 \[(\rho_0,\sigma_0) \xrightarrow{u_0} (\rho_1,\sigma_1) \xrightarrow{u_1}  \cdots  \xrightarrow{u_{k-1}}(\rho_k,\sigma_k),\]
 where $\xrightarrow{u}$  denotes a legal routing at vertex $u \in V_0$,
 is {\bf maximal}  if for all $u \in V_0$ we have $\sigma_k(u) \leq 0$, i.e. no other legal routing can be applied.
 
 The classic version of the commutativity result for rotor routing is the following:

 \begin{proposition}[\cite{Holroyd2008}]
    For all $(\rho,\sigma) \in \cR \times \Sigma$ with $\sigma \geq 0$, there is a unique $(\rho',\sigma')$ with $\sigma'(u)=0$ for all $u \in V_0$, such that all maximal legal routings from $(\rho,\sigma)$ end in $(\rho',\sigma')$. Furthermore, all legal routings can be continued in such a maximal legal routing.
 \end{proposition}

 The previous result states that we can always route legally all particles to the sinks, in any order by choosing every time a vertex where the routing is legal, and we will always reach the same final configuration. Moreover, it can be shown that the routing vectors corresponding to all maximum legal routings are the same. For such a maximal legal routing, we shall say that $(\rho,\sigma)$ is {\bf fully routed} to sinks, and write 
 \[(\rho',\sigma') = \routing^\infty_L(\rho,\sigma),\]
 where the $L$ stands for legal.

 \vskip .3cm
 
 The original {\sc arrival} problem consists in the following decision problem: 
\emph{if $(\rho, \sigma) \in \cR \times \Sigma$ with $\sigma \geq 0$ and $\deg(\sigma)=1$, if  $(\rho',\sigma') = \routing^\infty_L(\rho,\sigma)$, for a given sink $s \in S_0$, does $\sigma'(s)=1$ ? }

 \vskip .3cm
 
 This problem is known to be in NP and co-NP, but the best algorithm known to this date (see \cite{gartner_et_al:LIPIcs.ICALP.2021.69}) has complexity $2^{O(\sqrt{|V|})})$ in the case of a simple graph.
We shall now generalize this problem to any number of positive and negative particles, and remove the legality assumption.

\subsection{Equivalence classes of Rotors}

\begin{definition}
    Two rotor-particle configurations $(\rho,\sigma)$ and $(\rho',\sigma')$ are said to be equivalent, which we denote by $(\rho,\sigma) \sim (\rho',\sigma')$, if there is a routing vector $r$ such that
    \[ \routing^r(\rho,\sigma) = (\rho',\sigma').\]
\end{definition}

It is easy to see that this defines an equivalence relation on $\cR \times \Sigma$.

\begin{definition}
    Two rotor configurations $\rho, \rho'$ are said to be equivalent, which we denote by $\rho \sim \rho'$, if there is $\sigma \in \Sigma$ such that
    \[ (\rho,\sigma) \sim (\rho',\sigma). \]
\end{definition}

In this case, the relation is true for any $\sigma \in \Sigma$, and it defines an equivalence relation on $\cR$.

\paragraph{Cycle Pushes.}

Suppose that $\rho \in \cR$ and let $C$ be a directed circuit in $G(\rho)$. The \textbf{positive cycle push} of $C$ in $\rho$ transforms $\rho$ into $\theta_C \circ \rho$; see  Figure~\ref{fig:cyclepush}. Similarly, if $C$ is a directed circuit in $G(\theta^{-1}\circ \rho)$, the \textbf{negative cycle push} transforms $\rho$ into $\theta^-_C \circ \rho$. A \textbf{sequence of cycle pushes} is a finite or infinite sequence of rotor configurations $(\rho_i)$ such that each $\rho_{i+1}$ is obtained from $\rho_i$ by a positive or negative cycle push.

\begin{figure}[ht!]
    \centering
    \begin{subfigure}[c]{0.43\linewidth}
    \centering
    \begin{adjustbox}{max totalsize={.8\textwidth}{0.7\textheight},center}
        \begin{tikzpicture}
    \node[shape=circle,draw=black] (A) at (0,0) {$u_0$};
    \node[shape=circle,draw=black] (A') at (-1.5,0) {};    

    \node[shape=circle,draw=black] (B) at (2,1.5) {$u_1$};
    \node[shape=circle,draw=black] (B') at (3.5,1.5) {};

    \node[shape=circle,draw=black] (C) at (2,-1.5) {$u_2$};
    \node[shape=circle,draw=black] (C') at (3.5,-1.5) {};

    \path [->, >=latex](C) edge (A);
    \path [->, >=latex](A) edge (B);
    \path [->, >=latex](B) edge  (C);

    \path [->, >=latex,dashed, very thick](A) edge (A');
    \path [->, >=latex,dashed, very thick](B) edge (B');
    \path [->, >=latex,dashed, very thick](C) edge (C');

 \end{tikzpicture}
        \end{adjustbox}
    \end{subfigure}
    ~\vline~
      \begin{subfigure}[c]{0.43\linewidth}
        \begin{adjustbox}{max totalsize={.8\textwidth}{0.7\textheight},center}
        \begin{tikzpicture}
    \node[shape=circle,draw=black] (A) at (0,0) {$u_0$};
    \node[shape=circle,draw=black] (A') at (-1.5,0) {};    

    \node[shape=circle,draw=black] (B) at (2,1.5) {$u_1$};
    \node[shape=circle,draw=black] (B') at (3.5,1.5) {};

    \node[shape=circle,draw=black] (C) at (2,-1.5) {$u_2$};
    \node[shape=circle,draw=black] (C') at (3.5,-1.5) {};

    \path [->, >=latex,dashed, very thick](C) edge (A);
    \path [->, >=latex,dashed, very thick](A) edge (B);
    \path [->, >=latex,dashed, very thick](B) edge  (C);

    \path [->, >=latex](A) edge (A');
    \path [->, >=latex](B) edge (B');
    \path [->, >=latex](C) edge (C');

 \end{tikzpicture}
        \end{adjustbox}
    \end{subfigure}

\caption{On the left figure, the initial rotor configuration $\rho$ is given by
full arcs while other arcs are dashed. Applying a positive cycle push on the circuit of $G(\rho)$, formed by the vertices $u_0, u_1$ and $u_2$, results in the configuration shown in the right figure. Note that, in this example, there exists a unique rotor order.  }
\label{fig:cyclepush}
\end{figure}
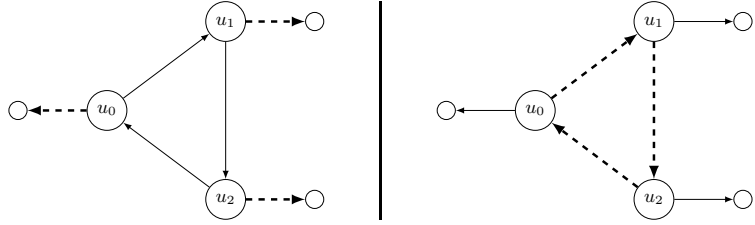

Note that if $C$ is a directed circuit in $G(\rho)$, for any $\sigma \in \Sigma$, we can obtain $(\theta_C \circ \rho, \sigma)$ by applying $\routing^{r_C}$ to $(\rho,\sigma)$, and if $C$ is a circuit in $G(\theta^{-1}\circ \rho)$, then $(\theta^{-1}_C \circ \rho,\sigma)$ is equal to $\routing^{-r_C} (\rho,\sigma)$, where in both cases $r_C$ is the routing vector consisting in routing once every vertex of $C$. In other words, a cycle push is a shortcut in the routing of a particle on the circuit.

\begin{theorem}
\label{thm.equivalence}
Given two rotor configurations $\rho$ and $\rho'$, $\rho \sim \rho'$ if and only if $\rho'$ can be obtained from $\rho$ by a sequence of cycle pushes.
\end{theorem}

\begin{proof}
Suppose that $\rho'$ can be obtained from $\rho$ by a sequence of cycle pushes.
Since cycle push operations can as well be obtained by routing operators, we have that for any $\sigma \in \Sigma$, $(\rho, \sigma) \sim (\rho', \sigma)$, and consequently $\rho \sim \rho'$.

Conversely, assume that, for a given $\sigma \in \Sigma$, there is a routing vector $r$ from $(\rho, \sigma)$ to $(\rho', \sigma)$. We show that $\rho'$ can be obtained by cycle pushes from $\rho$, by induction on the $L^1$-norm of $r$, i.e. $|r|_1=\sum_{u \in V_0} | r(u) |.$

If $|r|_1 = 0$, then $\rho = \rho'$.
Otherwise, consider the partition of $V$ in  sets $P$, $N$, and $Z$ corresponding to vertices $u$ such that $r(u)$ is positive, negative and null respectively. Assuming $P$ is nonempty (we can interchange the roles of $P$ and $N$ if needed), we observe that the degree of $\sigma$ on $P$, i.e. $\sum_{u \in P} \sigma(u)$, cannot increase through positive routing on $P$.   Similarly, negative routing on $N$ cannot increase the degree of $\sigma$ on $P$. However, after performing all the routings in $r$, we end up with the same particle configuration $\sigma$, which implies the degree on $P$ remains unchanged.  Consequently, all positive move operations within $P$ have exclusively been performed on arcs with head in $P$. 

In particular, $\rho(P)$ contains a directed circuit $C$. By applying a cycle push on circuit $C$, we obtain $(\theta_C \circ \rho, \sigma)$. Since a routing vector from 
$(\theta_C \circ \rho, \sigma)$ to $(\rho',\sigma')$ is
$r - r_C$, with $|r-r_C|_1 < |r|_1$,  we can apply induction to continue the sequence of cycle pushes.
\qed
\end{proof}

Whenever rotor configurations are equivalent, they eventually route particles identically since positive and negative cycle push correspond to adding or removing closed circuits in trajectories. In particular, it is easy to see that it is always possible to route any $(\rho,\sigma)$ to a $(\rho',\sigma')$ such that $\sigma'(u)=0$ for all $u \in V_0$. Let us denote by $\routing^\infty(\rho,\sigma)$ the nonempty set of these configurations.

\begin{theorem}
    Let $(\rho_1,\sigma_1) \in \routing^\infty(\rho,\sigma)$.
    Then $(\rho_2, \sigma_2) \in \routing^\infty(\rho,\sigma)$
    if and only if $\rho_1 \sim \rho_2$ and $\sigma_1=\sigma_2$.
\end{theorem}

\begin{proof}
    First, if $\rho_1 \sim \rho_2$ and $\sigma_1=\sigma_2$,
    by definition
    \[(\rho,\sigma) \sim (\rho_1,\sigma_1) \sim (\rho_2,\sigma_1) =  (\rho_2,\sigma_2), \]
    and $\sigma_2(u)=0$ for all $u \in V_0$,
    so that $(\rho_2,\sigma_2) \in \routing^\infty(\rho,\sigma)$.

    Conversely, suppose first that $|S_0| = 1$. Since $\deg(\sigma_1) = \deg(\sigma_2)$, one has $\sigma_1 = \sigma_2$, and by consequence $\rho_1 \sim \rho_2$.

    If $|S_0| > 1$, consider the rotor graph $G'$ obtained from $G$ by merging all sinks into a unique sink $s$. 
    Let $r$ be a routing vector  from $(\rho_1,\sigma_1)$ to $ (\rho_2,\sigma_2)$ in $G$. The same routing vector will also lead from $(\rho_1,\sigma'_1)$ to $ (\rho_2,\sigma'_2)$ in $G'$, where $\sigma'_1(u) = 0 $ for all $u \in V_0$ and $\sigma'_1(s) = \sum_{s' \in S_0} \sigma_1(s')$ (and $\sigma'_2$ defined accordingly). We deduce from the case $|S_0| = 1$  that $\sigma_1' = \sigma_2'$, and, by Theorem~\ref{thm.equivalence}, that $r$ corresponds to a sequence of cycle pushes in $G'$ and hence also in $G$. Since cycle push operations do not modify particle configurations, we have $\sigma_1 = \sigma_2$. 
\end{proof}

\begin{corollary} \label{cor:genarrival}
If $\sigma \geq 0$, then if $(\rho',\sigma') = \routing^\infty_L(\rho,\sigma)$ and $(\rho_1,\sigma_1) \in \routing^\infty(\rho,\sigma)$, we have $\sigma_1=\sigma'$ and
$\rho' \sim \rho_1$.
\end{corollary}

The {\sc generalized arrival} problem is:  \emph{given any $(\sigma,\rho)$, compute $\sigma_1$ for any $(\rho_1,\sigma_1) \in \routing^\infty(\rho,\sigma)$.}

Corollary \ref{cor:genarrival} shows that this problem contains the original {\sc arrival} problem. On the other hand, the decision version of {\sc generalized arrival} belongs to NP and co-NP, a certificate being a routing vector $r$; one may compute efficiently the configuration $\routing^r(\rho,\sigma)$ and check that we obtain 0 particles on $V_0$.

\subsubsection{Acyclic configurations}

We say that $\rho \in \cR$ is {\bf acyclic} if $G(\rho)$ contains no directed cycles. It amounts to saying that the set of arcs $\rho(V_0)$ forms in $G$ a directed forest, rooted in the sinks of $G$.

\begin{proposition}[\cite{giacaglia2011local}]
\label{prop.acyclic}
    Each equivalence class of rotor configurations contains exactly one acyclic configuration.
\end{proposition}

We can deduce from this result that the number of equivalence classes of rotor configurations is the number of rooted forests in $G$. By Kirchoff's  Matrix-Tree Theorem~\cite{pitman2018tree}, this is exactly the determinant of the Laplacian matrix 
of $G$ where we remove lines and columns corresponding to sinks; it also follows that this is the cardinal of the Sandpile Group of $G$ (see~\cite{Holroyd2008} and \ref{sec:sandpile}).

\subsection{Equivalence classes of particles}

\begin{definition}
    Two particle configurations $\sigma, \sigma'$ are said to be equivalent, which we denote by $\sigma \sim \sigma'$, if there is $\rho \in \cR$ such that
    \[ (\rho,\sigma) \sim (\rho,\sigma'). \]
\end{definition}

In this case, the relation is true for any $\rho \in \cR$, and it defines an equivalence relation on $\Sigma$. 

Define the Laplacian operator $\Delta$ as the linear operator from $\Z^{V_0}$ to $\Sigma$, defined for $u \in V_0$ by
\[ \Delta(u) = \sum_{a \in A^+(u)} (\head(a) - \tail(a))\]
The vector $\Delta(u)$, when added to a particle configuration $\sigma$, corresponds to transferring a total of $\deg^+(u)$ particles from $u$ to every outneighbour of $u$. The transformation from $\sigma$ to $\sigma + \Delta(u)$ is called ${\bf firing}$ $\sigma$ at $u$. This firing is legal if $\sigma(u) \geq \deg^+(u)$.

A {\bf firing vector} is simply an element of $r \in \Z^{V_0}$, and we can fire simultaneously vertices according to this vector by
\[\sigma + \Delta(r) = \sigma + \sum_{u \in V_0} r(u) \Delta(u).\]

\begin{proposition}
    For any two particle configurations $\sigma, \sigma'$ we have $\sigma \sim \sigma'$ if and only if there exists a firing vector $r$ with
    \[\sigma' = \sigma + \Delta(r).\]
\end{proposition}

\begin{proof}
    Let $r$ be a routing vector from $(\rho,\sigma)$ to $(\rho,\sigma')$.
    It follows that for all $u \in V_0$ we have $\theta^{r(u)}(\rho(u)) = \rho(u)$, so that $r(u)$ must be a multiple of $\deg^+(u)$ and we can write $r(u) = r'(u) \cdot \deg^+(u)$ with $r'(u) \in \Z$. From this follows that
    \[\sigma = \sigma + \Delta(r').\]
    Conversely, firing $\sigma$ at $u$ corresponds to $\deg^+(u)$ routings at $u$, which leaves the rotor configuration unchanged.
    \qed
\end{proof}

By analogy with maximal legal routings, define a maximal legal firing as a sequence of legal firings from $\sigma$ to another particle configuration $\sigma'$ such that finally $\sigma'$ is {\bf stable}, meaning that $\sigma'(u) < \deg^+(u)$ for all $u \in V_0$, i.e. no more legal firing are possible.

\begin{proposition}[\cite{bjorner1992chip}] \label{prop:stablechip}
    If $G$ is stopping, for all particle configurations $\sigma$ there is a unique configuration $\sigma'$ such that every maximal sequence of legal firings leads to $\sigma'$, and every sequence of legal firings can be continued in such a maximal sequence (in particular, all legal sequences are finite).
\end{proposition}

This stable configuration $\sigma'$ is the {\bf stabilization} of $\sigma$ and denoted $\sigma^\circ$.

\subsection{Sandpile Group} \label{sec:sandpile}

We point out that the equivalence relation on particles defined in the previous section is not equivalent to the construction of the so-called Sandpile Group.
In the case of a stopping rotor graph, the Sandpile Group is obtained from particle configurations equivalence classes by furthermore identifying configurations which have the same value on $V_0$. More precisely, define a relation $\sim_S$ by 

\[\sigma \sim_S \sigma' \Leftrightarrow \exists \sigma_1, \sigma \sim \sigma_1 \text{ and } \forall u \in V_0, \sigma'(u) = \sigma_1(u). \]
It is equivalent to requiring the existence of a
firing vector $r$ such that
 \[ \forall u \in V_0, \sigma'(u) = (\sigma + \Delta(r))(u).\]

\begin{proposition}[\cite{Holroyd2008}]
\begin{itemize}
    \item The quotient of $\Sigma$ by $\sim_S$ has an additive structure inherited from $\Sigma$, and it
    is a finite abelian group called the Sandpile Group
    and denoted by $SP(G)$;
    \item the order of $SP(G)$ is equal to the number of acyclic rotor configurations in $G$.
\end{itemize}    
\end{proposition}

\section{Main Results for Path Multigraphs}
In this part, we summarize our results, and the rest of the paper  will introduce the tools used to prove them. From now on, we consider  only graphs of the family $P^{x,y}_n$, and the letter $G$ denotes such a graph.

\subsection{The case $x=y=1$}
First, let us recall the results obtained about Path Graphs $P^{1,1}_n$ in \cite{auger2022polynomial} in order to understand how they compare to the case $P^{x,y}_n$ when $0<x<y$ are coprime. Technically, these results were stated only for nonnegative particle configurations but they still hold in the general case.

In the case $x=y=1$, define for any particle configuration $\sigma$
    \[h(\sigma) = \sum_{i=0}^{n+1} i \cdot \sigma(u_i)\]
    and for any rotor configuration $\rho$, define $g(\rho)$ as 
    \[g(\rho) = | i : \head(\rho(u_i)) = u_{i-1} |, \]
    i.e. $g(\rho)$ is the number of arcs in $G(\rho)$ pointing
    to the left.
    
    The next result completely solves {\sc generalized arrival} in $P^{1,1}_n$
    for any number of particles and antiparticles.

\begin{theorem} \label{thm:11arrival}
In the case $x=y=1$,
    for all $(\rho,\sigma) \in \cR \times \Sigma$, the number of particles on sink $u_{n+1}$ in any configuration of $\routing^\infty(\rho,\sigma)$
    is equal to the unique $m \in \Z$ such that 
    \[ 0 \leq g(\rho) -h(\sigma) + m (n+1) \leq n,\]
    i.e.
    \[  m = \lceil \frac{h(\sigma) - g(\rho)}{n+1} \rceil.\]
\end{theorem}

Together with this result, we can describe the structure of the Sandpile Group of $P^{1,1}_n$ and its action on rotor configurations.
Define $\bar{h}$ and $\bar{g}$  as $h$ and $g$ modulo $n+1$.

\begin{theorem} \label{thm:11sandpiles}
\begin{enumerate}[(i)]
    \item The Sandpile Group $SP(P^{1,1}_n)$ is cyclic of order $n+1$;
    \item the map $\bar{h} : \Sigma \rightarrow \Z/(n+1)\Z$ quotients by $\sim_S$ into an isomorphism between $SP(P^{1,1}_n)$ and $\Z/(n+1)\Z$;
    \item the map $\bar{g} : \cR \rightarrow \Z/(n+1)\Z$ quotients into a bijection between rotor equivalence classes  and $\Z/(n+1)\Z$;
    \item the action of the sandpile group on rotor equivalence classes can be understood in the following way: let $(\rho,\sigma)$ be a rotor-particle configuration and $(\rho',\sigma') \in \routing^\infty(\rho,\sigma)$.  Then $\rho'$ is in class 
    \[\bar{g}(\rho') = \bar{g}(\rho) - \bar{h}(\sigma) .\]
\end{enumerate}
\end{theorem}

As an example, consider the case $P^{1,1}_3$, which is depicted on Fig.~\ref{fig.example_simple}, with the particle configuration $\sigma$ equal to $(-8,5,10,-5,12)$ from left to right and $\rho$ as depicted. We see that $\rho$ has $2$ arcs going left so that $g(\rho)=2$, while we have
\[h(\sigma) = -8\cdot 0 +5\cdot 1+10\cdot 2 - 5\cdot 3 + 12\cdot 4 = 58.     \]

From Thm.~\ref{thm:11arrival}, we deduce the final configuration $\sigma'$ of the full routing of $(\rho,\sigma)$ counts $m=14$ particles ending on the right sink $u_4$ and $-8+5+10-5+12-14=0$ particles on $u_0$.

From Thm.~\ref{thm:11sandpiles}, we deduce that any final rotor configuration $\rho'$ in the routing will be such that $\bar{g}(\rho') = 2 - 58 = 0 \mod 4$,  so that all its arcs will point right, hence $\rho'$ is the acyclic configuration of this class.

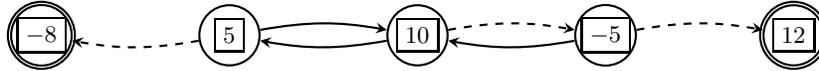
\begin{figure}[htbp]
\centering
\begin{tikzpicture}[>=stealth, auto, node distance=2.5cm, thick]

\tikzset{
  state/.style={circle, draw, inner sep=0pt, minimum size=8mm}
}

\node[state] (0) {$\boxed{-8}$};
\node[state, right of= 0] (1) {$\boxed{5}$};
\node[state, right of= 1] (2) {$\boxed{10}$};
\node[state, right of= 2] (3) {$\boxed{-5}$};
\node[state,  right of= 3] (4) {$\boxed{12}$};

\node[shape=circle,draw=black] (00) at (0) {$~~~~~~$};
\node[shape=circle,draw=black] (44) at (4) {$~~~~~~$};

\draw[->, dashed] (1) to[bend left=10]  (0);
\draw[->,thick] (1) to[bend left=10]  (2);

\draw[->,thick] (2) to[bend left=10]  (1);
\draw[->, dashed] (2) to[bend left=10]  (3);

\draw[->,thick] (3) to[bend left=10]  (2);
\draw[->, dashed] (3) to[bend left=10] (4);

\end{tikzpicture}
\caption{Rotor routing a particle configuration in $P^{1,1}_3$. The particle configuration is written in squares in vertices and the initial rotor configuration is given by full arcs while other arcs are dashed. Note that there is a unique rotor order in this graph.}
\label{fig.example_simple}
\end{figure}

\subsection{Case $0 < x < y$ coprime}

We now state our results in the case this paper is concerned about. Compare this with Theorem  \ref{thm:11arrival}.
In both theorems, we use
\[F = \sum_{i=0}^n x^{n-i} y^i.\]

\begin{theorem} \label{thm:arrivalcoprime}
Suppose that $0 < x < y$ are coprime and consider the rotor multigraph $P^{x,y}_n$.
\begin{enumerate}[(i)]
    \item There exists a linear function 
\[h : \Sigma \rightarrow \Z\]
and a function
\[g : \cR \rightarrow \Z\]
such that, for all $(\rho,\sigma) \in \cR \times \Sigma$, the number of particles on sink $u_{n+1}$ in any configuration of $\routing^\infty(\rho,\sigma)$
is equal to $m$ if and only if
\[g(\rho) -h(\sigma) +m F \in g(\cR); \]

\item the set $g(\cR)$ is a finite set of nonnegative integers, and membership in $g(\cR)$ can be tested in linear time; moreover the unique integer $m $ satisfying the previous condition can be found in time $O(n \log x)$, and it satisfies
\[ m - \lceil  \frac{h(\sigma) - g(\rho)}{F} \rceil \in \ll 0, x-1 \rr .\]
\item  More generally, if $(\rho,\sigma)$ and $(\rho',\sigma')$ are rotor-particle configurations, then  $(\rho,\sigma) \sim (\rho',\sigma')$ if and only if
    \[g(\rho) - h(\sigma) = g(\rho') - h(\sigma') \text{ and } \deg(\sigma)=\deg(\sigma').\]
\end{enumerate}
\end{theorem}

Note that, in the case $x=1$, we have 
\[ m = \lceil  \frac{h(\sigma) - g(\rho)}{F} \rceil \]
as in the case $x=y=1$ and no further algorithm is needed.

This is now the version of Theorem~\ref{thm:11sandpiles} in our present case. We define $\bar{h}$ and $\bar{g}$ as equal respectively to $h$ and $g$ modulo $F$.

\begin{theorem}
\label{thm.recap}
    Suppose that $0 < x < y$ are coprime and consider the rotor multigraph $P^{x,y}_n$.
    \begin{enumerate}[(i)]
        \item The Sandpile Group of $P^{x,y}_n$ is cyclic of order $F$;
        \item The map $\bar{h} : \Sigma \rightarrow \Z/F\Z$ quotients by $\sim_S$ into an isomorphism between $ SP(P^{x,y}_n)$ and $\Z/F\Z$;
        \item The map $\bar{g}$ quotients by $\sim$ into a bijection between rotor equivalence classes  and $\Z/F\Z$;
        \item The action of the sandpile group on rotor equivalence classes can be understood in the following way: let $(\rho,\sigma)$ be a rotor-particle configuration and $(\rho',\sigma') \in routing^\infty (\rho,\sigma)$.  Then $\rho'$ is in class 
        \[\bar{g}(\rho') = \bar{g}(\rho) - \bar{h}(\sigma) .\]
    \end{enumerate}
\end{theorem}

As an example, we consider the Path Multigraph $P^{2,3}_3$. The graph is depicted on Fig.~\ref{fig.gh}, together with harmonic values (values of $h$, inside vertices) and arcmonic values (values of $g$, on arcs).

Consider for instance the particle configuration $\sigma = (-8,5,13,-5,12)$ from left to right  such that 
\[h(\sigma) = -8 \times 0 + 5 \times 8 + 13 \times 20 - 5 \times 38 + 12 \times 65 = 890,\]
and the rotor configuration $\rho = (a_1^1, a^2_1, a^3_1)$ such that $g(\rho) = 12 +18 +27 = 57$.
We have $F=65$, and 
$g(\cR) = \{ 
0, 8, 12, 16, 18, 20, 24, 26,\allowbreak 27, 28, 30, 32,\allowbreak 34, 35, 36, \allowbreak38, 39, 40,\allowbreak 42, 43, \allowbreak 44, 45, 46,\allowbreak 47, 48, \allowbreak50, 51, 52, \allowbreak53, 54, 55, 56, 57,\allowbreak 58, 59, 60, 61, \allowbreak 62, \allowbreak 63, \allowbreak 64, \allowbreak 66,\allowbreak 67, \allowbreak 68, 69, 70, 71,\allowbreak 72, 74, 75, 76, 78, 79, 80,\allowbreak 82, 84, 86, 87, 88, \allowbreak 90, 94, 96, 98, \allowbreak 102, 106, \allowbreak 114 \}$.

The only value $v$ in $g(\cR)$ equal to $g(\rho)-h(\sigma)=-833 \mod 65$ is $12 = -833 + 13*65$. Since $\deg(\sigma)=17$, in the end of the routing there are $13$ particles on sink $u_4$ and $4$ particles on sink $u_0$. The final rotor configuration $\rho'$ satisfies
\[\bar{g}(\rho') = \bar{g}(\rho) - \bar{h}(\sigma) = -833 \mod 65 = 12 \mod 65\]
so $g(\rho')=12$ by looking in $g(\cR)$.

\begin{figure}[htbp]
\centering
\begin{tikzpicture}[>=stealth, auto, node distance=2.5cm, thick]

\tikzset{
  state/.style={circle, draw, inner sep=0pt, minimum size=8mm}
}

\node[state,text=red] (01) {$0$};
\node[state,text=red, right of= 01] (1) {$8$};
\node[state,text=red, right of= 1] (2) {$20$};
\node[state,text=red, right of= 2] (3) {$38$};
\node[state,text=red, right of= 3] (41) {$65$};

\node[shape=circle,draw=black] (0) at (01) {$~~~~~~$};
\node[shape=circle,draw=black] (4) at (41) {$~~~~~~$};
\draw [ thick,->,>=stealth](8,2) arc (0:330:0.4cm);

\draw[->] (1) to[bend left=10] node[above,text=blue] {24} (0);
\draw[->] (1) to[bend left=40] node[above,text=blue] {16} (0);
\draw[->] (1) to[bend left=80] node[above,text=blue] {8} (0);
\draw[->] (1) to[bend left=20] node[above,text=blue] {0} (2);
\draw[->] (1) to[bend left=60] node[above,text=blue] {12} (2);

\draw[->] (2) to[bend left=10] node[above,text=blue] {36} (1);
\draw[->] (2) to[bend left=40] node[above,text=blue] {24} (1);
\draw[->] (2) to[bend left=80] node[above,text=blue] {12} (1);
\draw[->] (2) to[bend left=20] node[above,text=blue] {0} (3);
\draw[->] (2) to[bend left=60] node[above,text=blue] {18} (3);

\draw[->] (3) to[bend left=10] node[above,text=blue] {54} (2);
\draw[->] (3) to[bend left=40] node[above,text=blue] {36} (2);
\draw[->] (3) to[bend left=80] node[above,text=blue] {18} (2);
\draw[->] (3) to[bend left=20] node[above,text=blue] {0} (4);
\draw[->] (3) to[bend left=60] node[above,text=blue] {27} (4);

\end{tikzpicture}
\caption{Harmonic and arcmonic values on $P^{2,3}_3$. The values of
$h$ and $g$ are given respectively in vertices and on arcs.}
\label{fig.gh} 
\end{figure}
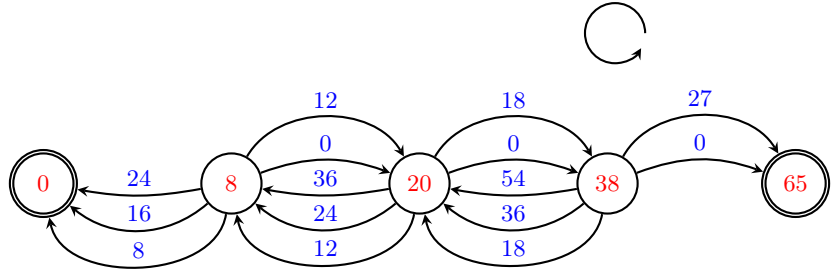

\section{Harmonic and Arcmonic Functions in the Path}

In the rest of the paper, we fix $n>0$ and coprime integers $x,y$ such that $0 < x < y$,  and consider the Path Multigraph $P_n^{x,y}$ as defined in Subsection~\ref{sec.graph_def}.

First, let us define the linear function $h : \Sigma \rightarrow \Z$, which will serve as an invariant for the firing operation and enable the characterization of particle equivalence classes. Initially, we define $h$ on  vertices and then extend it by linearity to $\Sigma$.

\begin{lemma}
The linear function $h:\Sigma \rightarrow \Z$ defined by  $h(u_0)=0$ and \[h(u_k)=\sum_{i=0}^{k-1}x^{n-i}y^{i}\]
for $k\in\llbracket1,n+1\rrbracket$ is harmonic on $G$, i.e. for any $u \in V_0$ we have
\[h(\Delta(u)) = 0.\]
\end{lemma}

\begin{proof}
For $k\in\ll 1,n \rr$: 
\[
y(h(u_{k})-h(u_{k-1}))=y(x^{n-k+1}y^{k-1})=x^{n-k+1}y^{k}
\]
 and 
\[
x(h(u_{k+1})-h(u_{k}))=x(x^{n-k}y^{k})=x^{n-k+1}y^{k}.
\]
Hence,
\begin{align*}
& y(h(u_{k})-h(u_{k-1}))  =x(h(u_{k+1})-h(u_{k}))\\
\Leftrightarrow \quad & (x+y)h(u_{k}) - yh(u_{k-1})-xh(u_{k+1})  = 0\\
\Leftrightarrow \quad & h(\Delta(u_k))  = 0.
\end{align*}
\qed
\end{proof}

\begin{corollary}
    For any particle configurations $\sigma, \sigma'$, if $\sigma \sim \sigma'$ then $h(\sigma) = h(\sigma')$.
\end{corollary}

 It turns out that $h(u_k)$ is the number of acyclic configurations in $P^{x,y}_n$ that contain a directed path from $u_k$ to $u_{n+1}$. In particular, $h(u_{n+1})$ is the number of rooted forests, which is also the number of particle equivalence classes and rotor equivalence classes~\cite{Holroyd2008}.

 In the rest of the document, we denote by $F$ this value, i.e.
 \[F = \sum_{i=0}^n x^{n-i} y^i.\]

We now define a similar function for rotor configurations, designed to be invariant on equivalence classes of rotors configurations. We introduce the term arcmonic for these functions that correspond to harmonic functions but on arcs.

\begin{proposition}
\label{prop.g_invariant}
 The linear function $g : \Z^A \rightarrow \Z$, defined by
\[g(a^k_j) = \sum_{i=0}^{j-1} (h(\head(a^k_i)) - h(u_k) )\]
for all $k \in \ll 1, n \rr$ and $j \in \ll 0,x+y-1\rr$ (in particular, $g(a^k_0)=0$)
is arcmonic, i.e. it satisfies for all directed circuits $C$ in $G(\rho)$, 
$g(C)=g(\theta(C)),$ where $C$ is identified with the sum of arcs $\sum_{a \in C} a$.
\end{proposition}

\begin{proof}
If $j \in  \ll 0,x+y-2\rr$ then
\begin{align*} 
g(\theta(a^k_j)) - g(a^k_j) & = g(a^k_{j+1}) - g(a^k_j) \\
& = \sum_{i=0}^j (h(\head(a^k_i)) - h(u_k) ) - \sum_{i=0}^{j-1} (h(\head(a^k_i)) - h(u_k) ) \\ 
& = h(\head(a^k_j)) - h(u_k)
\end{align*}

If $j=x+y-1$, then we use the fact that $h$ is harmonic so that $\sum_{i=0}^{j} (h(\head(a^k_i)) - h(u_k) )  = 0$,
\begin{align*}
    g(\theta(a^k_j)) - g(a^k_j) & = g(a^k_0) - g(a^k_j) \\
    & = \sum_{i=0}^{j} (h(\head(a^k_i)) - h(u_k) ) - \sum_{i=0}^{j-1} (h(\head(a^k_i)) - h(u_k) )\\
    & =  h(\head(a^k_j)) - h(u_k).
\end{align*}

Then, for any directed circuit $C$:
\begin{align*}
g(\theta(C)) - g(C) & =\sum_{a \in C} \left(h(\head(a)) - h(\tail(a)) \right) \\
& = 0.
\end{align*}
\qed
\end{proof}

By identifying a rotor configuration $\rho$ with the formal sum of its arcs, we can define 
\[g(\rho) = \sum_{u \in V_0} g(\rho(v)).\]

\begin{corollary}
\label{cor.rho_value}
If $\rho, \rho'$ are rotor configurations such that $\rho \sim \rho'$, then $g(\rho) = g(\rho')$.
\end{corollary}

The exact values of $g$ are given by:

\begin{proposition}
\label{prop.arcmonic_expr}
For $j\in\ll 0,x+y-1\rr$ and $k\in \ll 1,n \rr $, 
\[
g(a^k_j)=\begin{cases}
jd_{k} & \text{if }j\in\ll0,x\rr\\
(x+y-j)d_{k-1} & \text{if }j\in\ll x+1,x+y-1\rr
\end{cases}
\]
where, for every $k\geq 0$, $d_k = x^{n-k}y^k$.
\end{proposition}

Remark that, for every $k \in \ll 0,n  \rr$:
\[d_k = h(u_{k+1}) - h(u_k).\]

See Fig.~\ref{fig.gh} for an example of harmonic and arcmonic values on $P^{2,3}_3$.  In this example, $d_0 = 8$, $d_1 = 12$, $d_2=18$, $d_3 = 27$, and $d_4 = \frac{81}{2}$.

\begin{proposition}
\label{prop.gh_characterize}
    If $(\rho,\sigma)$ and $(\rho',\sigma')$ are rotor-particle configurations, then  if $(\rho,\sigma) \sim (\rho',\sigma')$ 
    we have
    \[g(\rho) - h(\sigma) = g(\rho') - h(\sigma').\]
\end{proposition}

\begin{proof}
    Without loss of generality, assume that $(\rho', \sigma') = \routing^+_u(\rho, \sigma)$ for some $u\in V_0$.  Recall that, by definition of routing operators, $\sigma' = \sigma + \head(\rho(u))-u $, hence by the linearity of $h$ we obtain:
    \[ h(\sigma') - h(\sigma) = h(\head(\rho(u))) - h(u) .\] 

    We use the fact that $g(\rho')-g(\rho) = (h(\head(\rho(u))) - h(u))$ (see the proof of Prop.~\ref{prop.g_invariant}).
    This yields:
    \begin{align*}
     g(\rho') - h(\sigma') - (g(\rho) - h(\sigma)) & = (g(\rho') - g(\rho)) + (h(\sigma) - h(\sigma')) \\
    & = (h(\head(\rho(u))) - h(u)) + (h(\sigma) - h(\sigma'))\\
    & = 0.
    \end{align*}
    \qed
\end{proof}

\subsection{Stable decomposition of arcmonic values}

In the light of Prop.~\ref{prop.gh_characterize}, it becomes important to characterize which integers are of the form $g(\rho)$ for some $\rho \in \cR$.
If $\rho \in \cR$, by Proposition~\ref{prop.arcmonic_expr}, $g(\rho)$ can be decomposed as a sum
\[g(\rho) = \sum_{k=0}^n c_k d_k,\]
with $c_k \in \ll 0,x+y-1\rr$ for all $k \in \ll 1,n\rr$; recall that $d_k = x^{n-k} y^k$.

This decomposition is not unique since all equivalent rotor configurations share the same value. We shall show that to each equivalence class we can assign a special form of decomposition, named \textbf{stable decomposition} thereafter.

\begin{theorem}
\label{thm.decomposition_unique}
    Every integer $v \geq 0$ has unique decomposition of the form
    \[ v = \sum_{k=0}^n c_k d_k + c_{n+1} d_{n+1}\]
    with $c_k \in \ll0,y-1\rr$ for $k \in \ll0,n\rr$
    and $c_{n+1} \in x\Z$.
\end{theorem}

Note that $d_{n+1} = \frac{y^{n+1}}{x}$. A special case is the case  $x=1$ where if $v < y^{n+1}$, the stable decomposition of $v$ coincides with the decomposition of $v$ in base $y$ up to the $n$-th element. 

\begin{proof}
    We establish the uniqueness of this stable decomposition. The existence  relies on the lemmas presented subsequently.

    Suppose that $v$ admits two stable decompositions $c^1 = (c^1_0, \dots, c^1_{n+1})$ and $c^2 = (c^2_0, \dots, c^2_{n+1})$. Recall that, for $i \in \{1,2\}$, $c^i_{n+1} \in x\Z$. Then: \[\sum_{k=0}^n c^1_k d_k + c^1_{n+1} d_{n+1} = \sum_{k=0}^n c^2_k d_k + c^2_{n+1} d_{n+1} \mod y\]
    which amounts to 
    \[c^1_0 d_0 = c^2_0 d_0 \mod y.\]
    Since $d_0=x^n$ and $y$ are coprime, and $0 \leq c^1_0, c^2_0 \leq y-1$, we obtain $c^1_0 = c^2_0$. Now, consider $v' = \frac{v - c^1_0 x^n}{y}$, then, for $i \in \{1,2\}$,
    \[v' = \sum_{k=0}^{n-1} c^i_{k+1} x^{n-1-k}y^k + c^i_{n+1} \frac{y^n}{x}  \]
    and one can apply the same reasoning iteratively on $v'$ to show that $c_1^1 = c_1^2$, $c_2^1 = c_2^2$, etc. And finally that $c^1 = c^2$.
    \qed
\end{proof}

 To prove the existence of the stable decomposition, we rely on another device named \textbf{Engel Machine}~\cite{engel1975probabilistic}.

The Engel Machine $E^{x,y}_n$ is the Multigraph defined
on the set $\{u_0,u_1,\allowbreak\cdots,\allowbreak u_n\}\allowbreak \cup \{u_{n+1},s\}$, where every vertex $u_i$ for $i \in \ll 0,n\rr$ has $x$ arcs going to $u_{i+1}$ and $y-x$ arcs going to $s$. Since we assumed $y > x$, then $y-x>0$. Vertices $s$ and $u_{n+1}$ are sinks.
We say that a particle  configuration $\sigma$ in $E^{x,y}_n$ is \textbf{nonnegative} if $\sigma(u_i) \geq 0$ for $i \in \ll 0, n \rr$ (whereas sinks may have a negative value). See Fig. \ref{fig:ex_engel} for an example.

\begin{figure}[htbp]
     \centering
\begin{tikzpicture}[>=stealth, auto, node distance=2cm, thick]

\tikzset{
  state/.style={circle, draw, inner sep=0pt, minimum size=8mm}
}

\node[state] (01) {$u_0$};
\node[state] (1) [right of=01] {$u_1$};
\node[state] (2) [right of=1] {$u_2$};
\node[state] (3) [right of=2] {$...$};
\node[state] (4) [right of=3] {$u_n$};
\node[state] (51) [right of=4] {$u_{n+1}$};
\node[state] (s) [below of=2] {$s$};

\node[shape=circle,draw=black] (s0) at (s) {$~~~~~~$};
\node[shape=circle,draw=black] (5) at (51) {$~~~~~~$};


\path[->]
  (01) edge[bend left=30] (1)
  (01) edge[bend left=50] (1)
  (01) edge[bend right=60] (s)
  
  (1) edge[bend right=20] node[left] {$y-x$ arcs}  (s)
  (1) edge[bend left=30] (2)
  (1) edge[bend left=50] (2)

  (2) edge[bend left=20] (s)
  (2) edge[bend left=30] (3)
  (2) edge[bend left=50] node[above] {$x$ arcs} (3)

  (3) edge[bend left=20] (s)
  (3) edge[bend left=30] (4)
  (3) edge[bend left=50] (4)

  (4) edge[bend left=20] node[right] {$y-x$ arcs}  (s)
  (4) edge[bend left=30] (5)
  (4) edge[bend left=50] (5);

\end{tikzpicture}
\caption{\footnotesize The Engel Machine $E^{2,3}_n$.}
\label{fig:ex_engel} 
\end{figure}
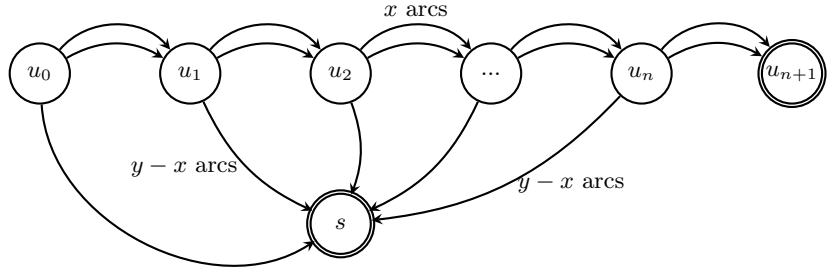

We define a function $h_E$ on the vertices of this graph that will turn out to be harmonic on $E_n^{x,y}$. This function is defined by
\[h_E(s)=0 \text{ and } h_E(u_k) = d_k \text{ for k in } \ll 0,n+1\rr\]
and extend it to particle configurations by linearity.

We shall be mainly concerned with the $h_E$ value of particle configurations in the Engel Machine. In order to keep notation simple, and since $h_E(s)=0$, the value of configurations on $s$ never matters and we identify particle configurations in $c \in  \Sigma(E^{x,y}_n)$ with words $c \in \Z^{n+2}$. In particular, for any $v \geq 0$, the notation $c[v]$ denotes the word corresponding to the stable decomposition of $v$, as well as a (stable) particle configuration (we can suppose that its value on $s$ is always $0$). Note that $h_E(c[v]) = v$ by construction.
Conversely, remark that any nonnegative \textit{stable} configuration $c$ with $h_E(c)=v$ gives the unique \textit{stable} decomposition of $v$.

\begin{lemma} 
    \label{lem.hE_harmonic}
    The function $h_E$ is harmonic on $E^{x,y}_n$.
\end{lemma}

\begin{proof}
Consider the particle configuration $c'$ obtained from $c$ by firing vertex $u_k$, $k \in \ll 0, n\rr$. Then:
\[ h_E(c')-h_E(c) = - y d_k + x d_{k+1} =0 .\]
\qed
\end{proof}

In order to compute a stable decomposition for $v$, one simply has to find any configuration $c$ with $h_E(c)=v$ and then stabilize $c$. The proof of the next lemma provides a method for computing such a configuration $c$.  Together with Lemma \ref{lem.hE_harmonic}, this completes the proof of Theorem~\ref{thm.decomposition_unique}.

\begin{lemma} For any $v \geq 0$, there exists a nonnegative configuration $c$ in $E^{x,y}_n$ with $h_E(c)=v$.
\end{lemma}

\begin{proof}
    Since $x^{n+1}$ and $y^{n+1}$ are coprime, by Bezout's theorem there are integers $\alpha, \beta$ such that
    \[\alpha x^{n+1} + \beta y^{n+1} = 1\]
    and we can choose $\alpha \geq 0$.
    It follows that
    \[(\alpha x v) x^{n+1} + (\beta x v) y^{n+1} = x v\]
    and
    \[(\alpha x v) d_0 + (\beta x v) d_{n+1} = v. \]
    \qed
\end{proof}

\subsection{Recognizing decompositions of arcmonic values}

In this subsection, we characterize stable decompositions corresponding to an arcmonic value.

\begin{theorem}
    \label{thm.v_decomposition}
    For any $v \in \Z$, we have $v \in g(\cR)$ if and only if the regular expression 

    \[ e_d = \ll 0, y-1\rr^* \cdot 0 \cdot \ll 1, x \rr^* \cdot  0 \]
    
    matches $c[v]$. 
\end{theorem}

The proof is split in several lemmas. We define the regular expressions
\[ e_a = \ll 1,y \rr^* \cdot 0 \cdot \ll 0,x-1 \rr^*\cdot 0,\]
and 
    \[ e_d = \ll 0, y-1\rr^* \cdot 0 \cdot \ll 1, x \rr^* \cdot  0 .\]
Let $L_a$ and $L_d$ be the languages described by $e_a$ and $e_d$ respectively, and let $L_a^n$ and $L_d^n$ be the subsets of words of length $n+2$.

\begin{lemma}
   \label{lem.decompose_tree}
   There is a bijective function $\psi$  between the set of acyclic rotor
   configurations of $P^{x,y}_n$, and $L^n_a$, such that
   for any acyclic rotor configuration $\rho$, we have
   \[g(\rho) = h_E( \psi(\rho) ).\]
\end{lemma}

\begin{proof}
Let $\rho$ be an acyclic configuration of $P^{x,y}_n$. For such a configuration, there is some $k \in \ll 1,n \rr$ such that
\begin{itemize}
    \item for $i < k$, $\rho(u_i)=a^i_j$ with $j \in \ll x,x+y-1 \rr$ and
    $g(a^i_j) = c_{i-1} d_{i-1}$ with $c_{i-1} \in \ll 1, y \rr$,
    \item  for $i \geq k$, $\rho(u_i)=a^i_j$ with $j \in \ll 0,x-1 \rr$ and
    $g(a^i_j) = c_i d_i$ with $c_i \in \ll 0, x-1 \rr$.
\end{itemize}
If we define $c_{k-1}=0$ and $c_{n+1}=0$, the configuration $c=(c_0,c_1,\dots,c_n,c_{n+1})$
satisfies $h_E(c) = g(\rho)$ and is matched by $e_a$; we define $\psi(\rho) = c$.

Conversely, for any configuration $c$ matched by $e_a$ it is easy to see that
there is a unique acyclic configuration $\rho$ with $\psi(\rho)=c$; if $k-1$ is the position of first $0$ of $c$,
we can construct $\rho$ as above.
\qed
\end{proof}

\begin{lemma}
\label{lem.language}
If $c \in L_a$, let $\phi(c)$ be the stable decomposition of $c$. Then $\phi$ defines a bijective map between $L_a^n$ and $L_d^n$ that preserves $h_E$.
\end{lemma}

\begin{proof}
By definition, if $c \in L^n_a$, then $\phi(c)=c^\circ$ is the stabilization of $c$.
By Lemma~\ref{lem.hE_harmonic}, $h_E$ is harmonic, hence $h_E(c) =   h_E(\phi(c))$.

We introduce a sequential transducer $T$, depicted on Fig. \ref{fig.transducer}, which computes the stabilization of certain configurations. The notation $\ll a , b \rr | \ll a + k , b + k\rr$ represents the substitution of any integer $i$ in $\ll a , b \rr$ by the integer $i+k$. This transducer takes as input any word in $\ll0,y\rr^*$ and produces a word in $\ll0,y-1\rr^*$ with the same length. In particular, when given a nonnegative configuration $c$ of $E_n^{x,y}$, satisfying $c(u) \leq y$ for all $u \in V_0$ and $c(u_{n+1}) = 0$, the transducer outputs stabilized configuration $\phi(c)$ (recall
that we do not record what happens on sink $s$), stabilizing $c$ from vertex $u_0$ to $u_n$ in ascending order. Hence it computes $\phi$ for configurations in $\ll 0,y \rr^* \cdot 0$.

\begin{figure}[htbp]
\begin{center}
\begin{tikzpicture}[shorten >=1pt, initial text=,node distance=2cm, on grid, auto, node distance = 4cm]

    \node[state,initial,accepting]    (a)               {$a$};
    \node[state,accepting]            (b)  [right=of a] {$b$};

    \path[->]
    (a) edge [loop above] node {$\ll 0,y-1\rr | \ll0,y-1\rr$} (a)
    (a) edge [bend left] node {$y | 0$} (b)
    (b) edge [loop right] node {$\ll y-x,y\rr | \ll 0,x\rr$} (b)
    (b) edge [bend left] node {$\ll 0,y-x-1\rr | \ll x,y-1\rr$} (a)
    ;
\end{tikzpicture}
\end{center}
\caption{Transducer that computes $\phi(c)$, the stabilization of a particle configuration $c$, for any $c \in \ll 0,y \rr^* \cdot 0$. Notation $\ll a , b \rr | \ll a + k , b + k\rr$ stands for the substitution of any integer $i$ in $\ll a , b \rr$ by $i+k$.}
\label{fig.transducer}
\end{figure}
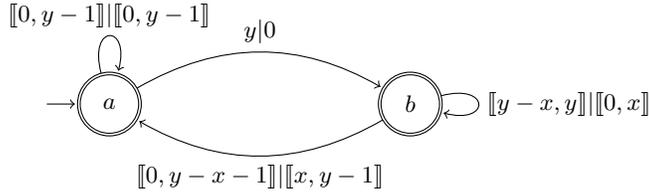

Consider now the automaton $A_a$ depicted on Fig. \ref{fig.automateA}, which recognizes the language $L_a$. 

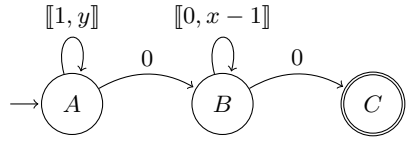
\begin{figure}[htbp]
\begin{center}
\begin{tikzpicture}[shorten >=1pt, initial text=,node distance=2cm, on grid, auto]

    \node[state,initial]    (a)               {$A$};
    \node[state]            (b)  [right=of a] {$B$};
    \node[state,accepting]  (c)  [right=of b] {$C$};

    \path[->]
    (a) edge [loop above] node {$\ll 1,y\rr$} (a)
    (a) edge [bend left] node {$0$} (b)
    (b) edge [loop above] node {$\ll 0,x-1\rr $} (b)
    (b) edge [bend left] node {$0$} (c)
    ;
\end{tikzpicture}
\end{center}
\caption{Automaton $A_a$ recognizing $L_a$}
\label{fig.automateA}
\end{figure}

From $T$ and $A_a$, we build the transducer $T(A_a)$ depicted on Fig. \ref{fig.transducercompose} which is the product of $T$ and $A$. Given a configuration $c$ in $\ll0,y\rr^*\cdot \Z$, the product transducer accepts it if and only if $c \in L_a$, and in such case outputs $\phi(c)$.

\begin{figure}[htbp]
\begin{center}
\begin{tikzpicture}[shorten >=1pt, initial text=,node distance=2cm, on grid, auto]

    \node[state,initial]    (Aa)               {$Aa$};
    \node[state]            (Ba)  [node distance = 4cm, right=of Aa] {$Ba$};
    \node[state,accepting]  (Ca)  [right=of Ba] {$Ca$};
    \node[state]            (Ab)  [node distance = 3cm, below =of Aa] {$Ab$};

    \path[->]
    (Aa) edge [loop above] node {$\ll 1,y-1\rr | \ll 1,y-1\rr$} (Aa)
    (Aa) edge  node [below left] {$0 | 0$} (Ba)
    (Aa) edge [bend left] node [left] {$y | 0$} (Ab)
    (Ab) edge [loop below] node {$\ll y-x,y\rr | \ll 0,x\rr$} (Ab)
    (Ab) edge [bend left] node {$\ll 1,y-x-1\rr | \ll x+1,y-1\rr$} (Aa)
    (Ab) edge  node [below right] {$0 | x$} (Ba)
    (Ba) edge [loop above ] node {$\ll 0,x-1\rr | \ll 0,x-1\rr$} (Ba)
    (Ba) edge  node {$0 | 0$} (Ca)
    ;
\end{tikzpicture}
\end{center}
\caption{The product $T(A_a)$ of transducer $T$ and automaton $A_a$}
\label{fig.transducercompose}
\end{figure}
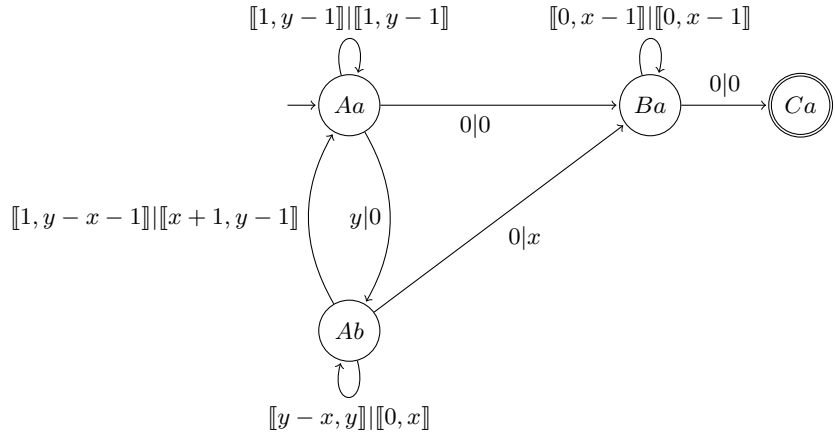

From $T(A_a)$, if we look at the output of every transition as an input, we get an automaton which recognizes exactly $\phi(L_a)$. This automaton is depicted on Fig.~\ref{fig.ndaL_d}, and its determinization on Fig.~\ref{fig.daL_d}

\begin{figure}[htbp]
\begin{center}
\begin{tikzpicture}[shorten >=1pt, initial text=,node distance=2cm, on grid, auto]

    \node[state,initial]    (Aa)               {$Aa$};
    \node[state]            (Ba)  [node distance = 4cm, right=of Aa] {$Ba$};
    \node[state,accepting]  (Ca)  [right=of Ba] {$Ca$};
    \node[state]            (Ab)  [node distance = 3cm, below =of Aa] {$Ab$};

    \path[->]
    (Aa) edge [loop above] node {$\ll 1,y-1\rr$} (Aa)
    (Aa) edge  node [below left] {$0$} (Ba)
    (Aa) edge [bend left] node [left] {$0$} (Ab)
    (Ab) edge [loop below] node {$\ll 0,x\rr$} (Ab)
    (Ab) edge [bend left] node {$\ll x+1,y-1\rr$} (Aa)
    (Ab) edge  node [below right] {$x$} (Ba)
    (Ba) edge [loop above ] node {$\ll 0,x-1\rr$} (Ba)
    (Ba) edge  node {$0$} (Ca)
    ;
\end{tikzpicture}
\end{center}
\caption{The nondeterministic automaton obtained for $\phi(L_a)$}
\label{fig.ndaL_d}
\end{figure}
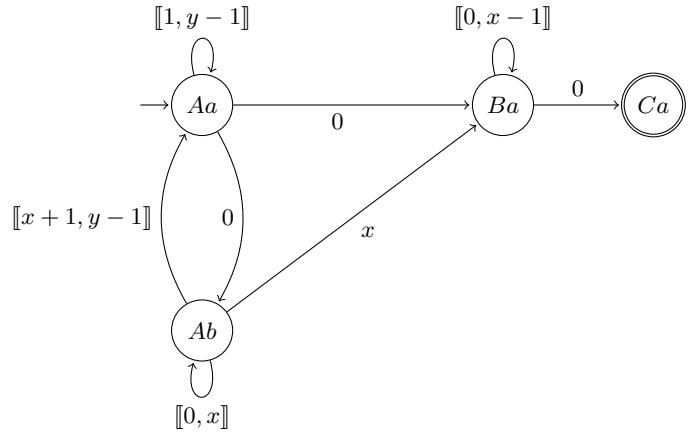

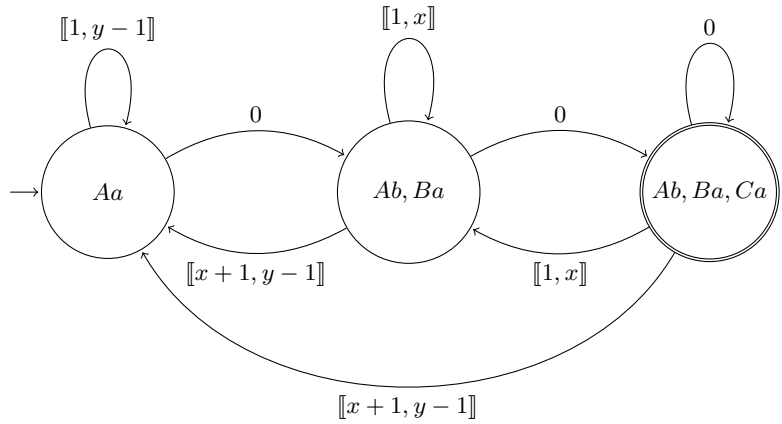
\begin{figure}[htbp]
\begin{center}
\begin{tikzpicture}[shorten >=1pt, initial text=,node distance=2cm, on grid, auto, node distance = 4cm]

    \node[state,initial]    (a)               {$~~~~~Aa~~~~~$};
    \node[state]            (b)  [right=of a] {$~~~Ab,Ba~~~$};
    \node[state,accepting]  (c)  [right=of b] {$Ab,Ba,Ca$};

    \path[->]
    (a) edge [loop above] node {$\ll 1,y-1\rr$} (a)
    (a) edge [bend left] node {$0$} (b)
    (b) edge [loop above] node {$\ll 1,x\rr $} (b)
    (b) edge [bend left] node {$0$} (c)
    (b) edge [bend left] node {$\ll x+1, y-1 \rr$} (a)
    (c) edge [bend left] node {$\ll 1, x \rr$} (b)
    (c) edge [bend left=60] node {$\ll x+1, y-1 \rr$} (a)
    (c) edge [loop above] node {$0$} (c)

    ;
\end{tikzpicture}
\end{center}
\caption{Determinization of the automaton from Fig. \ref{fig.ndaL_d}. This automaton is minimal. }
\label{fig.daL_d}
\end{figure}

It is now easy to check that the automaton for $\phi(L_a)$ depicted on Fig.~\ref{fig.daL_d} recognizes exactly $L_d$,
since this automaton is minimal.
Moreover, as $\phi$ preserves the length of words, we deduce that $\phi(L_a^n) = L_d^n$.  Additionally both languages $L^n_a$ and $L_d^n$ have the same size, namely $F$. It follows that $\phi$ is a bijective map between $L^n_a$ and $L_d^n$. 
\qed
\end{proof}

\begin{proof}[of Theorem~\ref{thm.v_decomposition}]

The value $v$ belongs to $g(\cR)$ if and only if there is an acyclic
configuration $\rho$ such that $v = h_E(\psi(\rho)) = h_E(\phi(\psi(\rho))) = h_E(c[v])$ , (Lemma~\ref{lem.decompose_tree}),
hence if and only if $c[v] \in L_d^n$ (Lemma~\ref{lem.language}).
\qed
\end{proof}

The uniqueness of the stable decomposition  together with the previous result implies:
\begin{proposition}
\label{cor.g_unique}
For $\rho, \rho' \in \cR$, we have $\rho \sim \rho'$ if and only if $g(\rho)=g(\rho')$.
\end{proposition}

\begin{proof}
    The forward direction was proved as Corollary~\ref{cor.rho_value}.
    
    Conversely, by the same corollary we can suppose that $\rho$ and $\rho'$ are acyclic and satisfy $g(\rho)=g(\rho')$.
    It follows by Lemma~\ref{lem.decompose_tree} that $g(\rho)=h_E(\psi(\rho))$
    and $g(\rho')=h_E(\psi(\rho'))$; then 
    \[h_E( \phi(\psi(\rho)) ) = h_E( \phi(\psi(\rho')) ).\]
    By uniqueness of the stable decomposition, it follows
    that $\phi(\psi(\rho))=\phi(\psi(\rho'))$,
    and since $\phi$ and $\psi$ are bijective, that $\rho=\rho'$.
    \qed
\end{proof}

\begin{lemma}
    \label{cor.monotony}
    For any value $v \geq 0$, the value of $c[v+kF](u_{n+1})$ is nondecreasing with $k$.
\end{lemma}

\begin{proof}
    Let $c_1, c_2$ be two nonnegative configurations in $E^{x,y}_n$, 
    define $c = (c_1 + c_2)^\circ$. Then
    by considering the stabilization mechanism, we have
    \[ c(u_{n+1}) \geq c_1(u_{n+1}) + c_2(u_{n+1})\]
    from which the result follows.
    \qed
\end{proof}

\begin{lemma}
\label{lem.monotony}
Let $v$ be an integer. Then:
\begin{enumerate}
    \item[$(i)$]  If $c[v](u_{n+1}) < 0$, then $v-F \notin g(\cR)$.
    \item[$(ii)$] If $c[v](u_{n+1}) \geq 0$, then $v+F \notin g(\cR)$.
\end{enumerate}
\end{lemma}

\begin{proof}
    Notice that $c[F]=(1,1,\dots,1,0)$.
    
    $(i)$: if  $v$ is such that $c[v](u_{n+1}) < 0$, then by Lemma~\ref{cor.monotony}  $c[v - F](u_{n+1}) < 0$. Hence it is not matched by the regular expression $e_d$ of Theorem~\ref{thm.v_decomposition}.

    $(ii)$: by the same argument, if $c[v](u_{n+1}) > 0$, then $c[v+F](u_{n+1}) > 0$ and $v+F \notin g(\cR)$.

    Consider now $v$ such that $c[v](u_{n+1}) = 0$. Then $c[v]$ corresponds to a word of length $n+2$ in the language  $L = \ll 0, y-1 \rr^* \cdot 0$. Let us consider $c_1$ such that $c_1(u_i) = c(u_i) + 1$ for all $i \in \ll 0, n\rr$ and $c_1(u_{n+1}) = 0$, so that $h_E(c_1) = h_E(c[v+F])$. Now, we aim to demonstrate that the regular expression $e_d$ in Theorem~\ref{thm.v_decomposition} does not match the stable decomposition of $v+F$ computed from $c_1$. We do this by relying on the construction of an automaton that recognizes the set of possible stable decompositions of $c_1$ for all possible $c$.

    Recall the notation:
    \[ e_d = \ll 0, y-1\rr^* \cdot 0 \cdot \ll 1, x \rr^* \cdot  0 ,\]
     $L_d$ is the language described by  $e_d$, and  $L_d^n$ is the subset of words of length $n+2$. 
    Moreover $\phi(c)$ is the stable decomposition of $c$ for any $c \in \ll 0, y \rr^* \cdot 0$, while $\phi$ is computed by the transducer $T$ described in Figure~\ref{fig.transducer}.

    The set of possible configurations $c_1$ when
    $c$ varies in the set of stable configurations with $c(u_{n+1})=0$, is exactly described by the regular expression $\ll 1, y \rr^* \cdot 0$, corresponding to a language $L_1$, which is recognized by the automaton $A_1$ depicted on Fig.~\ref{fig.A1}. 

\begin{figure}[htbp]

    \begin{center}
    \begin{tikzpicture}[shorten >=1pt, initial text=,node distance=2cm, on grid, auto]

    \node[state,initial]    (a)               {$A$};
    \node[state, accepting] (b)  [right=of a] {$B$};

    \path[->]
    (a) edge [loop above] node {$\ll 1,y\rr$} (a)
    (a) edge [bend left] node {$0$} (b)
    ;
    \end{tikzpicture}
    \end{center}
    \caption{Automaton $A_1$ recognizing $L_1$.}
    \label{fig.A1}
\end{figure}
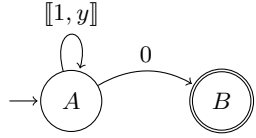

    Following the steps of the proof of Lemma~\ref{lem.language}, and since words in $L_1$ are also matched by  $\ll 0, y \rr^* \cdot 0$, we construct the product transducer $T(A_1)$ which outputs $\phi(c)$ if and only if $c \in L_1$. See Fig.~\ref{fig.T_A1}.

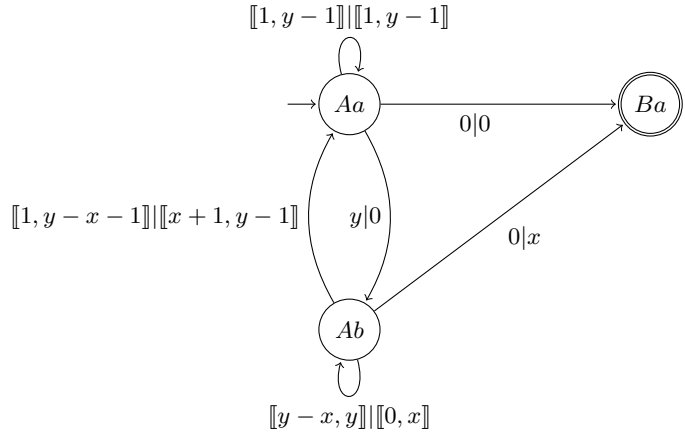
\begin{figure}[htbp]
    \begin{center}
    \begin{tikzpicture}[shorten >=1pt, initial text=,node distance=2cm, on grid, auto]

    \node[state,initial]    (Aa)               {$Aa$};
    \node[state,accepting]   (Ba)  [node distance = 4cm, right=of Aa] {$Ba$};
    \node[state]            (Ab)  [node distance = 3cm, below=of Aa] {$Ab$};

    \path[->]
    (Aa) edge [loop above] node {$\ll 1,y-1\rr | \ll 1,y-1\rr$} (Aa)
    (Aa) edge  node [below left] {$0 | 0$} (Ba)
    (Aa) edge [bend left] node [left] {$y | 0$} (Ab)
    (Ab) edge [loop below] node {$\ll y-x,y\rr | \ll 0,x\rr$} (Ab)
    (Ab) edge [bend left] node {$\ll 1,y-x-1\rr | \ll x+1,y-1\rr$} (Aa)
    (Ab) edge  node [below right] {$0 | x$} (Ba)
    ;
    \end{tikzpicture}
    \end{center}
    \caption{The product $T(A_1)$ of transducer $T$ and automaton $A_1$.}
    \label{fig.T_A1}
\end{figure}

    Finally, the following non-deterministic automaton $A_1^\phi$ recognizes $\phi(L_1)$ as shown on Fig.~\ref{fig.phi_L1}.

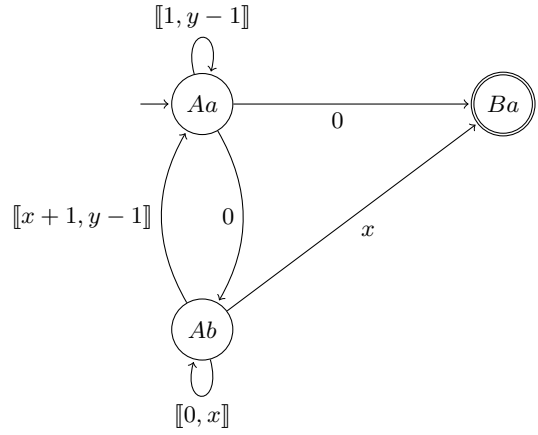
\begin{figure}

    \begin{center}
    \begin{tikzpicture}[shorten >=1pt, initial text=,node distance=2cm, on grid, auto]

    \node[state,initial]    (Aa)               {$Aa$};
    \node[state,accepting]  (Ba)  [node distance = 4cm, right=of Aa] {$Ba$};
    \node[state]            (Ab)  [node distance = 3cm, below =of Aa] {$Ab$};

    \path[->]
    (Aa) edge [loop above] node {$\ll 1,y-1\rr$} (Aa)
    (Aa) edge  node [below left] {$0$} (Ba)
    (Aa) edge [bend left] node [left] {$0$} (Ab)
    (Ab) edge [loop below] node {$\ll 0,x\rr$} (Ab)
    (Ab) edge [bend left] node {$\ll x+1,y-1\rr$} (Aa)
    (Ab) edge  node [below right] {$x$} (Ba)
    ;
    \end{tikzpicture}
    \end{center}
    \caption{The non-deterministic $A^\phi_1$ automaton that recognizes $\phi(L_1)$.}
    \label{fig.phi_L1}
\end{figure}

    It suffices to show that this automaton does not recognize any word in $L_d$ or, equivalently, that $\phi(L_1) \cap L_d = \emptyset$. To that end, we introduce the automaton $A_d$ that recognizes $L_d$ on Fig.~\ref{fig.Ad}.

\begin{figure}[htbp]
    \begin{center}
    \begin{tikzpicture}[shorten >=1pt, initial text=,node distance=4cm, on grid, auto]

    \node[state,initial]    (a)               {$\alpha$};
    \node[state]            (b)  [right=of a] {$\beta$};
    \node[state,accepting]  (c)  [right=of b] {$\gamma$};

    \path[->]
    (a) edge [loop above] node {$\ll 1,y-1\rr$} (a)
    (a) edge [bend left] node {$0$} (b)
    (b) edge [loop above] node {$\ll 1,x\rr $} (b)
    (b) edge [bend left] node {$0$} (c)
    (b) edge [bend left] node {$\ll x+1, y-1 \rr$} (a)
    (c) edge [bend left] node {$\ll 1, x \rr$} (b)
    (c) edge [bend left=60] node {$\ll x+1, y-1 \rr$} (a)
    (c) edge [loop above] node {$0$} (c)
    
    ;
    \end{tikzpicture}
    \end{center}
    \caption{Automaton $A_d$ that recognizes $L_d$.}
    \label{fig.Ad}
\end{figure}
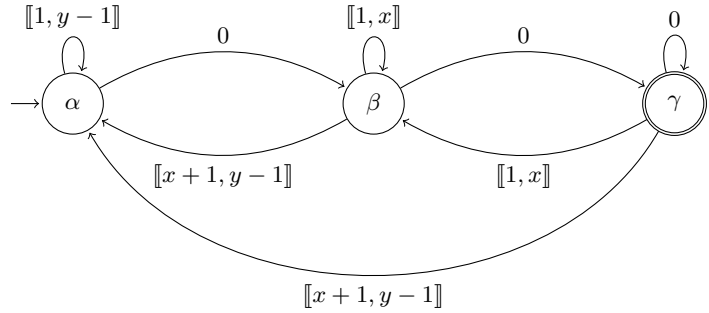

    By taking the product of automatas $A_d$ and $A_1^\phi$, we obtain an automaton that recognizes $\phi(L_1) \cap L_d$, as shown on Fig.~\ref{fig.prod_phi_1}.  This automaton does not contain any accepting state which proves that $\phi(L_1) \cap L_d = \emptyset$.

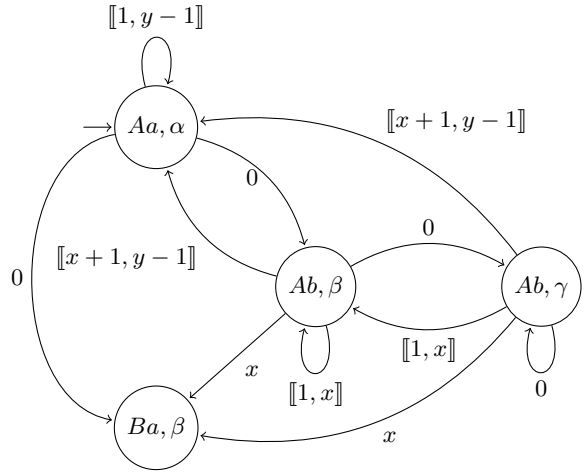
\begin{figure}[htbp]

    \begin{center}
    \begin{tikzpicture}[shorten >=1pt, initial text=,node distance=2cm, on grid, auto]

    \node[state,initial]    (a)               {$Aa,\alpha$};
    \node[state]  (b)  [node distance = 4cm, below=of a] {$Ba,\beta$};
    \node[state]  (c)  [node distance = 3cm, below right =of a] {$Ab,\beta$};
    \node[state]  (d)  [node distance = 3cm, right =of c] {$Ab,\gamma$};

    \path[->]    
    (a) edge [loop above] node {$\ll 1, y-1 \rr$} (a)
    (a) edge  [bend right=80] node [left] {$0$} (b)
    (a) edge [bend left] node [left] {$0$} (c)
    (c) edge [loop below] node {$\ll 1,x \rr$} (c)
    (c) edge [bend left] node {$\ll x+1, y-1 \rr$} (a)
    (c) edge  node {$x$} (b)
    (c) edge [bend left] node  {$0$} (d)
    (d) edge [loop below] node {$0$} (d)
    (d) edge [bend right] node [above right] {$\ll x+1, y-1\rr$} (a)
    (d) edge [bend left] node {$x$} (b)
    (d) edge [bend left] node {$\ll 1,x \rr$} (c)  
    ;
    \end{tikzpicture}
    \end{center}
    \caption{The product of automata $A_d$ and $A^\phi_1$ that recognizes $\phi(L_1) \cap L_d$; it does not contain any final state.}
    \label{fig.prod_phi_1}
\end{figure}

    \qed  
\end{proof}

The next Lemma is key to proving our main results and helps in improving the complexity of our algorithm.

\begin{lemma}
\label{thm.unique_value_modulo_F}
For every $0 \leq v \leq F-1$, there is a unique $k \in \mathbb{N}$ such that $v+kF \in g(\cR)$, which is the smallest integer $k$ with $c[v+kF](u_{n+1}) \geq 0$.
\end{lemma}

\begin{proof}
The uniqueness is a consequence of  Lemma~\ref{lem.monotony} and the monotony of $c[v+kF](u_{n+1})$ with $k$. As stated in Proposition~\ref{cor.g_unique}, the function $g$ uniquely identifies  rotor classes. Hence, the existence of $k$ such that $v+kF \in g(\cR)$ follows from the observation that the number of rotor classes is precisely $F$. In other words, the function $g \mod F$ establishes a bijective correspondence between the set of rotor classes and $\Z / F\Z$.
\qed
\end{proof}

As an example, consider $P^{2,3}_3$ as depicted in Fig.~\ref{fig.gh}, and value $v = 1$. Next table shows the stable decomposition of $v+kF$, with $F = 65$, for $k \in \ll 0 , 3 \rr$. The unique value in $g(\cR)$ is $66$ whose stable decomposition is matched by the regular expression of Theorem~\ref{thm.v_decomposition}.

\begin{center}
\begin{tabular}{c|c}
\hline
$~~~k~~~$ & $~~~$ stable decomposition of $1+65k~~~$ \\
\hline
\hline
0 & $(2,1,0,2,-2)$ \\
1 &  $(0,1,0,2,~~0)$\\
2 &  $(1,2,1,0,~~2)$\\
3 &  $(2,0,1,0,~~4)$\\
\end{tabular}
\end{center}

\section{Proofs of  Theorem  \ref{thm:arrivalcoprime} and \ref{thm.recap} }

\paragraph{Proof of Theorem  \ref{thm:arrivalcoprime}}

\paragraph{(i):}
If $(\rho',\sigma') \in \routing^\infty(\rho,\sigma)$, then by Proposition \ref{prop.gh_characterize} we have
\begin{align*}
        g(\rho) - h(\sigma) & = g(\rho') - h(\sigma')\\
        g(\rho) - h(\sigma) + h(\sigma') & =  g(\rho').
\end{align*}
Since $\sigma'$ is zero, except on $u_0$ and $u_{n+1}$
where the value of $h$ is respectively $0$ and $F$, we get
  \[g(\rho) - h(\sigma) + mF \in g(\cR) \]
  where $m=\sigma'(u_{n+1})$.

Conversely, suppose that there is another $m_1$ such that 
\[g(\rho) - h(\sigma) + m_1F = g(\rho_1).\]
for some $\rho_1$. Then 
\[m_1 F - g(\rho_1) = m F - g(\rho')\]
hence
\[ g(\rho_1) = g( \rho') \mod F .\]
By Lemma \ref{thm.unique_value_modulo_F}, it follows that $g(\rho_1)=g(\rho')$ hence $m_1=m$.

\paragraph{(ii):}

Recall that
\[ F = \sum_{k=0}^n d_k .\]
Since the maximal arcmonic value of an arc in $A^+(u_k)$ is $y d_{k-1} = x d_k$ for $k \in \ll 1,n\rr$, we obtain that  the maximal value in $g(\cR)$ is $\sum_{k=1}^n x d_k$ which is strictly lower than $xF$. Then:
\begin{align*}
         & 0  \leq g(\rho)-h(\sigma)+mF  < xF\\
        \Leftrightarrow &  h(\sigma) - g(\rho)  \leq  mF  < xF + h(\sigma) - g(\rho)\\
         \Leftrightarrow & \lceil  \frac{h(\sigma) - g(\rho)}{F} \rceil  \leq m  < x + \lceil  \frac{h(\sigma) - g(\rho)}{F} \rceil\\ 
        \Leftrightarrow & 0  \leq m - \lceil  \frac{h(\sigma) - g(\rho)}{F} \rceil \leq x -1 \\  
\end{align*}

If we are given $(\rho,\sigma)$ and $m$ and want to decide if there are $m$ particles on sink $u_{n+1}$ when fully routing $(\rho,\sigma)$, we can either check:
\begin{itemize}
    \item if $c[g(\rho)-h(\sigma)+mF]$ is matched by the regular expression $e_d$, which involves first
    computing the stable decomposition;
    \item if $c[g(\rho)-h(\sigma)+mF](u_{n+1})=0$ and $c[g(\rho)-h(\sigma)+(m-1)F)](u_{n+1}) < 0$,
    which involves computing two stable decompositions.
\end{itemize}

Assuming that elementary arithmetic operations are $O(1)$, 
we can compute $g(\rho) - h(\sigma) +mF$ in time $O(n)$, using Prop. \ref{prop.arcmonic_expr} for $g$.
Then, computing a stable decomposition also has computational complexity $O(n)$. We can successively fire all vertices from $u_0$ to $u_n$, which can be done by computing a quotient and
remainder modulo $x+y$.

If we are given $(\rho,\sigma)$ and we want to compute $m$, we can proceed by bissection, using  Lemma \ref{thm.unique_value_modulo_F} to find the minimal $m$ for which $c[g(\rho)-h(\sigma)+mF](u_{n+1})=0$. The overall complexity of this method is $O(n \log(x))$ since $m$ belongs to an interval of length $x$.

\paragraph{(iii):}

The forward direction is Prop.~\ref{prop.gh_characterize}.

Conversely, suppose that $g(\rho)-h(\sigma) = g(\rho')-h(\sigma')$. Let $m$ and $m'$
be the number of particles on sink $u_{n+1}$ when we fully route  $(\rho,\sigma)$ and $(\rho',\sigma')$ respectively  to sinks; we denote respectively the final configurations of these routings by $(\rho_1,\sigma_1)$ and $(\rho'_1,\sigma'_1)$. By Prop.~\ref{prop.gh_characterize},
\[g(\rho) - h(\sigma) = g(\rho') - h(\sigma') = g(\rho_1) - mF = g(\rho'_1) - m'F, \]
from which we deduce by Lemma~\ref{thm.unique_value_modulo_F} that $\rho_1 \sim \rho'_1$ and $m=m'$, so that 
$\sigma_1=\sigma'_1$ (since $\deg(\sigma_1)=\deg(\sigma'_1)$). All in all, we have that 
\[ (\rho,\sigma) \sim (\rho_1,\sigma_1) \sim (\rho'_1,\sigma_1) \sim (\rho',\sigma'). \]

\paragraph{Proof of Theorem \ref{thm.recap}}

\paragraph{(i) and (ii):}

Suppose that $\bar{h}(\bar{\sigma}_1) = \bar{h}(\bar{\sigma}_2)$.
Up to adding particles to $\bar{\sigma}_2$ on $u_{n+1}$ and on $u_0$ we obtain $\sigma_2$ such that $h(\bar{\sigma}_1) = h(\sigma_2)$ and $\deg(\bar{\sigma}_1) = \deg(\sigma_2)$ respectively. We write $\sigma_1 = \bar{\sigma}_1$.

Consider now any $\rho \in \cR$ . We have
\[h(\sigma_1) - g(\rho) = h(\sigma_2) - g(\rho), \]
so
\[(\rho,\sigma_1) \sim (\rho,\sigma_2) \]
by $(iii)$ of Theorem~\ref{thm:arrivalcoprime},
and $\bar{\sigma}_1 \sim_S \bar{\sigma}_2$.

Conversely, if $\bar{\sigma}_1 \sim_S \bar{\sigma}_2$, we clearly have 
$\bar{h}(\bar{\sigma}_1) = \bar{h}(\bar{\sigma}_2)$.

Since $\bar{h}(u_1)=x^n \mod F$ and $x^n$ is coprime with $F$, we see that the particle configuration with just one particle on $u_1$ generates all possible values in $\Z / F \Z$.
It follows, by the first isomorphism theorem, that $SP(P^{x,y}_n)$ is cyclic and isomorphic to $\Z / F \Z$.

\paragraph{(iii):} follows directly from Lemma \ref{thm.unique_value_modulo_F}.

\subsubsection{Open problems and future works} 

In this paper, we addressed the generalized version of the {\sc arrival} problem in the Path Multigraph  $P^{x,y}_n$. Moreover, we investigated the Sandpile Group structure and its action on rotor configurations when $x$ and $y$ are coprime. However, when $x$ and $y$ are not coprime, we observed that the characterization of classes by harmonic and arcmonic functions becomes inadequate, necessitating the inclusion of more comprehensive algebraic invariants. We are currently working on a project that presents a theory of arcmonic and harmonic functions applicable to general graphs, which will be submitted soon to publication.

Moreover, it is worth considering other scenarios, such as variations in $x$ and $y$ across different vertices or changes in the rotor order. These cases pose interesting questions that require further investigation. We regard them as open problems that warrant additional research.

\subsubsection{Acknowledgements} 
Thanks to Chloé  and Marwanne for checking examples with their rotor software.

This work was supported by a public grant as part of the Investissement d'avenir project, reference ANR-11-LABX-0056-LMH, LabEx LMH.

\end{document}